\documentclass[12pt,a4paper,draft,twoside,reqno]{amsart}

\usepackage{amssymb}
\usepackage{amsmath}
\usepackage{amsfonts}

\newtheorem{theorem}{Theorem}

\newtheorem{corollary}[theorem]{Corollary}

\newtheorem{proposition}[theorem]{Proposition}

\newcommand*{\R}{\mathbb R} 
\begin{document}

\qquad\qquad\qquad Dedicated to the 70th anniversary of Peter Olver
\bigskip

\bigskip

\title[Third order differential operators]{Natural differential invariants and equivalence of third order nonlinear differential operators}
\author[Valentin Lychagin]{Valentin Lychagin}
\address{Institute of Control Sciences of RAS, Moscow, Russia}
\email{Valentin.Lychagin@uit.no}
\author[Valeriy Yumaguzhin]{Valeriy Yumaguzhin}
\address{Program Systems Institute of RAS, Pereslavl'-Zales\-skiy, Russia} 
\email{yuma@diffiety.botik.ru}
\thanks{V. Yumaguzhin is the corresponding author}

\begin{abstract}
 We give a description of the field of rational natural differential
invariants for a class of nonlinear differential operators of the third
order on a two dimensional manifold and show their application to the equivalence problem of such operators.
\end{abstract}

\maketitle
\section{Introduction}

In this paper, we continue  to study rational differential invariants of differential operators on two-dimensional manifolds.  Here we consider a class of nonlinear operators 
that in local coordinates $x_1, x_2$ have the following form:
\begin{multline}\label{Aw}
A_{w}\colon f\longmapsto \sum\limits_{i,j,k} a_{ijk}(x,f) \frac{\partial
^{3}f}{\partial x_i\partial x_j\partial x_k}+\sum\limits_{i,j} a_{ij}\left( x,f\right) \frac{\partial
^{2}f}{\partial x_{i}\partial x_{j}}\\
+\sum\limits_{i,j}a_{i}\left( x,f\right) 
\frac{\partial f}{\partial x_{i}}+a_{0}\left( x,f\right)f,  
\end{multline}
where $x=( x_1, x_2)$, $f=f( x)$. 

We call such operators as \textit{weakly nonlinear} operators.

In paper \cite{LY2w}, we found rational differential invariants of the 2nd order weakly nonlinear operators and used them to solve the local equivalence problem. 

Here we use the methods of \cite{LY2w} to find rational differential invariants for the 3rd order  weakly nonlinear operators  and apply them to solve the local as well as the global equivalence problem.

Additionally we will assume that all coefficients $a(x,y)$ of these operators 
are smooth in $x$ and rational in $y$.

It is easy to see that the pseudogroup of local diffeomorphisms (in variables $x$) naturally acts in this class of operators. 

The main step in studies of such operators is a representation of operator $A_w$ as a pair 
$(A, f)$ (we call it as a ({\it related pair}), where 
\begin{multline}\label{A}
A=\sum\limits_{i,j,k}a_{ijk}\left( x,y\right) \frac{\partial ^{3}}{\partial
x_{i}\partial x_{j}\partial x_{k}}+\sum\limits_{i,j}a_{ij}\left( x,y\right) \frac{\partial ^{2}}{\partial
x_{i}\partial x_{j}}\\
+\sum\limits_{i}a_{i}\left( x,y\right) \frac{\partial}{\partial x_i}+a_0(x, y)  
\end{multline}%
is a linear third order linear differential operator on the extended space $(x, y)$.

The {\it procedure of descent} allows us to get differential invariants of weakly nonlinear operators from invariants of  related pairs and  linear
operators of the form 
\begin{multline}\label{A_f}
A_{f}=\sum\limits_{i,j,k}a_{ijk}\left( x,f\right) \frac{\partial
^{3}}{\partial x_{i}\partial x_{j}\partial x_{k}}+\sum\limits_{i,j}a_{ij}\left( x,f\right) \frac{\partial^{2}}{\partial x_{i}\partial x_{j}}\\
+\sum\limits_{i}a_{i}\left( x,f\right)\frac{\partial }{\partial x_{i}}+a_{0}\left( x,f\right).  
\end{multline}

Using this procedure and a machinery of reduced Gr\"{o}bner bases, we propose a method to find invariants of the 3rd order weakly nonlinear operators. 

\section{Notations}

In this paper, we use the same notations as  in \cite{LY3, LYk}. 

Let $M$ be $n$-dimensional manifold and let $\tau\colon TM\to M$ and \\$\tau^*\colon T^*M \to M$ be respectively tangent and cotangent bundles over $M$.  

We will denote by $\mathit{Diff}_k (M)$ the module of linear $k$-th order operators, acting in $C^{\infty}(M)$. The corresponding vector bundle of such operators we denote by 
$\chi_k\colon\mathbf{Diff}_k\to M$, thus $\mathit{Diff}_k(M)$ is also the module of smooth sections of $\chi_k$, $\mathit{Diff}_k(M)=C^{\infty}(\chi_k)$.

The group of diffeomorphisms of $M$ will be denoted by $\mathcal{G}(M)$ and the  multiplicative group of nowhere vanishing functions on $M$ will be denoted as $\mathcal{F}(M)$.

We denote by $\Sigma^k(M)$ and $\Sigma_k(M)$ the modules of symmetric $k$-forms and $k$-vectors respectively.

Also, we denote by $\Omega^k$ the modules of exterior $k$-forms.

\section{Linear differential operators on $M$}

In this section, we recall the necessary facts from \cite{LY3} on scalar linear third order differential operators on two-dimensional manifolds.

Let $M$ be 2-dimensional smooth oriented connected manifold and let
$$
  \pi\colon M\times\R\longrightarrow M
$$
be  trivial line bundle.
Denote by 
\begin{equation*}
\pi _{k}\colon\mathbf{J}^{k}(\pi) \longrightarrow M
\end{equation*}
the bundles of $k$-jets of smooth sections of $\pi$ (i.e. $k$-jets of smooth functions on $M$).

Let $\mathcal{G}(M)$ be the group of diffeomorphisms of the manifold $M$. 

Then, together with the action of $\mathcal{G}\left( M\right) $ on $M$, we have also the actions $\mathcal{G}\left( M\right) $ in the bundles $\pi _{k}$ by prolongations of diffeomorphisms. 

The prolongations of diffeomorphisms $\phi \in \mathcal{G}(M)$ in the bundles $\pi _{k}$ will be denoted by $\phi ^{(k)}$.

\subsection{} Let 
\begin{equation*}
\chi_3\colon\mathbf{Diff}_{3}\left( M\right) \longrightarrow M
\end{equation*}%
be the bundle of the third order linear differential operators on $M$ and 
\begin{equation*}
\chi_{3, k}\colon\mathbf{J}^{k}(\chi_3) \longrightarrow M
\end{equation*}
be bundles of their $k$-jets of sections of $\chi_3$.

The group $\mathcal{G}\left( M\right) $ acts on operators $A\in\mathit{Diff}_{3}(M)$ in natural way:%
$$
\phi_{\ast }\colon A \rightarrow \phi _{\ast }(A) =\phi_{*}\circ A \circ \phi _{*}^{-1},
$$
where the morphism $\phi_{\ast}\colon C^{\infty }(M) \longrightarrow C^{\infty}(M)$ is defined by the formula
$$
  \phi_{\ast }=(\phi ^{-1}) ^{\ast}\colon h\mapsto h\circ \phi ^{-1}.  
$$

\subsection{Symbols} By the {\it symbol} $\sigma_{3, A}$ of the operator $A$, we means the equivalence class 
$$
  \sigma_{3, A}= A\!\!\!\mod\!\mathit{Diff}_2(M).
$$

Let an operator $A\in\mathit{Diff}_3(M)$ be represented in local coordinates $x_1, x_2$ of $M$ in the form
\begin{multline}\label{A1}
A=a_1\partial^3_1+3a_2\partial^2_1\partial_2+3a_3\partial_1\partial^2_2+a_4\partial^3_2\\
+b_1\partial^2_1+2b_2\partial_1\partial_2+b_3\partial^2_2
+c_1\partial_1+c_2\partial_2+a_0,
\end{multline}
 where the all coefficients  are smooth functions of $x=(x_1, x_2)$, $\partial_1$ and 
 $\partial_2$ are $\partial_{x_1}$ and  $\partial_{x_2}$ respectively.

Then the symbol $\sigma_{3, A}$ is identified with the symmetric 3-vector
\begin{equation}\label{Smbl}
  \sigma_{3, A}=a_1\partial^3_1+3a_2\partial^2_1\partial_2+3a_3\partial_1\partial^2_2
  +a_4\partial^3_2\in\Sigma_3(M) ,
\end{equation}
where dots and degrees are symmetric products of the vector fields $\partial_1,\; \partial_2$. 

The symbol $\sigma_{3,A}$ is {\it regular} (at point or in domain) if it as a homogeneous cubic  polynomial on the cotangent bundle $T^*(M)$ has distinct roots. 

Denote by $\Delta(\sigma_{3, A})$  the discriminant of $\sigma_{3, A}$, then
$$
\Delta(\sigma_{3, A})=6a_1a_2a_3a_4-4(a_1a_3^3+a_4a_2^3)+3a_2^2a_3^2-a_1^2a_4^2.
$$ 
 
Recall that the symbol $\sigma_{3, A}$ has three distinct real roots if $\Delta(\sigma_{3, A})>0$ and one real and two complex roots if $\Delta(\sigma_{3, A})<0$. 

Thus, we say that the operator $A$ is {\it regular} if its symbol $\sigma_{3,A}$ is regular or $\Delta(\sigma_{3, A})\neq 0$

More over, we say that an operator $A$ is {\it hyperbolic} if  $\Delta(\sigma_{3, A})>0$ and {\it ultrahyperbolic}  if $\Delta(\sigma_{3, A})<0$.

Locally, the symbol of a hyperbolic operator can be presented as a symmetric product of pair wise  linear independent vector fields $\sigma_{3, A}=X_1\cdot X_2\cdot X_3$, and therefore there are local coordinates $x_1, x_2$ such that 
\begin{equation}\label{Smblh}
\sigma_{3, A}=(a\partial_1+b\partial_2)\cdot\partial_1\cdot\partial_2,
\end{equation}
where $a$ and $b$ are smooth functions and $ab\neq 0$. 

For the case of ultrahyperbolic  operators we have $\sigma_{3, A}=X\cdot q$, where $X$ is a nonzero vector field and $q\in\Sigma_2$ is a positive symmetric 2-vector. Therefore, there are local coordinates $x_1, x_2$ such that
\begin{equation}\label{Smblu}
\sigma_{3, A}=(a\partial_1+b\partial_2)\cdot(\partial^2_1+\partial^2_2),
\end{equation}
where $a$ and $b$ are smooth functions and $a^2+b^2>0$.

\subsection{Wagner connections}
The following result due to Wagner \cite{Wgnr, LY3}.
\begin{theorem} $\phantom{1}$
\begin{enumerate}
\item  Let $A\in {\it Diff}_3(M)$ be a regular differential operator. Then there exist a unique linear connection $\nabla$ such that the symbol $\sigma_{3, A}$ is parallel with respect to $\nabla$. 
\item The curvature tensor of $\nabla$ is zero.
\end{enumerate}
\end{theorem}

We call this connection {\it Wagner's connection}.

 In the case of hyperbolic symbol, we choose local coordinates in $M$ such that $\sigma$ has form \eqref{Smblh}. Then the non zero Christoffel coefficients $\Gamma^i_{jk}$ of the Wagner connection are the following:
\begin{gather*}
 \Gamma^1_{11}=\frac{1}{3}\big(\ln\frac{b}{a^2}\big)_{x_1},\quad
 \Gamma^2_{22}=\frac{1}{3}\big(\ln\frac{a}{b^2}\big)_{x_2},\\
 \Gamma^1_{12}=\frac{1}{3}\big(\ln\frac{b}{a^2}\big)_{x_2},\quad
 \Gamma^2_{21}=\frac{1}{3}\big(\ln\frac{a}{b^2}\big)_{x_1}.
\end{gather*}
In the ultrahyperbolic case \eqref{Smblu} we have the following nonzero Chritoffel coefficients:
\begin{gather*}
 \Gamma^1_{12}=\Gamma^2_{22}=-\frac{1}{6}\big(\ln(a^2+b^2)\big)_{x_2},\\
 \Gamma^1_{21}=-\Gamma^2_{11}=\frac{a_{x_1}b-ab_{x_1}}{a^2+b^2},\quad
 \Gamma^1_{22}=-\Gamma^2_{12}=\frac{ab_{x_2}-a_{x_2}b}{a^2+b^2}.
\end{gather*}

\begin{corollary}\label{Crllr2}\phantom{1}
The torsion form $\theta$ of the Wagner connection is 
$$
\theta=\frac{1}{3}\big(\ln\frac{b^2}{a}\big)_{x_1}dx_1+\frac{1}{3}\big(\ln\frac{a^2}{b}\big)_{x_2}dx_2
$$
for the hyperbolic case and
$$
\theta=\frac{ab_{x_2}-a_{x_2}b}{a^2+b^2}dx_1+\frac{ab_{x_1}-a_{x_1}b}{a^2+b^2}dx_2-\frac{1}{6}d\big(\ln(a^2+b^2)\big)
$$
for the ultrahyperbolic case.
\end{corollary} 

\subsection{Symols and quantization} 

Let $\Sigma^{\cdot} = \oplus_{k\ge 0}\Sigma^k(M)$ be the gra\-ded algebra of symmetric differential forms and let $\nabla$ be the Wagner connection associated with a regular symbol from $\Sigma_3(M)$. Then the covariant differential
$$
d_{\nabla}:\Omega^1(M)\longrightarrow \Omega^1(M)\otimes\Omega^1(M)
$$
define derivation 
$$
 d_{\nabla}^s\colon\Sigma^{\cdot}\longrightarrow\Sigma^{\cdot +1} 
$$
of degree one in graded symmetric algebra $\Sigma^{\cdot} \oplus_{k\ge 0}\Sigma^k(M)$. Namely, this derivation is defined by its action on generators, and we have
\begin{align*}
&d^s_{\nabla}= d : C^{\infty}(M)\longrightarrow\Omega^1(M) = \Sigma^1,\\
&d^s_{\nabla}:\Omega^1(M)= \Sigma^1\stackrel{d_{\nabla}}{\longrightarrow}\Omega^1(M)\otimes\Omega^1(M)\stackrel{\mathrm{Sym}}{\longrightarrow}\Sigma^2.
\end{align*}

Let now $\alpha_k\in\Sigma_k(M)$. We define a differential operator 
$\mathcal Q(\alpha_k)\in\mathit{Diff}_k(M)$ as follows:
$$ 
 \mathcal Q(\alpha_k)(h)\stackrel{\mathrm{def}}{=}\frac{1}{k!}\left\langle\,\alpha_k,\,\big(d^s_{\nabla}\big)^k(h)\,\right\rangle,
$$ 
where $h\in C^{\infty}(M)$, $\big(d^s_{\nabla}\big)^k(h)\in\Sigma^k(M)$, and $\langle\cdot\, ,\cdot\rangle$ is the standard convolution
$$
  \Sigma_k(M)\otimes\Sigma^k(M)\longrightarrow C^{\infty}(M).
$$

Remark, that the value of the symbol of the derivation $d^s_{\nabla}$
at a covector $\theta$ equals to the symmetric product by $\theta$ into the module $\Sigma^{\cdot}$.  Therefore, the symbol of operator $ \mathcal Q(\alpha_k)$ equals $\alpha_k$ as the symbol of a composition of operators equals the composition of symbols.

We call differential operator $ \mathcal Q(\alpha_k)$ a {\it quantization of symbol $\alpha_k$}.

Let now $A\in\mathit{Diff}_3 (M)$ and $\sigma_{3, A}$ be its symbol. Then operator
$$
A-\mathcal Q\big(\sigma_{3, A}\big)
$$
has order $2$, and let $\sigma_{2,A}$ be its symbol.

Then operator $A- \mathcal Q\big(\sigma_{3, A}\big)-\mathcal Q\big(\sigma_{2,A}\big)$ has order $1$ and let $\sigma_{1, A}$ be its symbol. Thus we get {\it subsymbols} $\sigma_{i, A}\in\Sigma_i(M)$, $0\le i\le 2$, such that
$$
A = \mathcal Q\big(\sigma_{(3)}(A)\big),
$$
where
\begin{equation}\label{Subsymbol} 
\sigma_{(3)}(A) =\sigma_{3, A}+\sigma_{2, A}+\sigma_{1, A} +\sigma_{0, A}
\end{equation} 
is the total symbol
and 
$$
\mathcal{Q}\big(\sigma_{(3)}(A)\big) = \mathcal Q\big(\sigma_{3, A}\big) +\mathcal Q\big(\sigma_{2, A}\big)+\ldots +\mathcal Q\big(\sigma_{0, A}\big).
$$

\subsubsection{Coordinates}

Let $x_1, x_2$ be local coordinates in a neighborhood $\mathcal O\subset M$, where the symbol $\sigma_{3, A}$ is regular. Denote by $x_1, x_2, w_1,w_2$ induced standard coordinates in the tangent bundle over $\mathcal O$.

Then $d_{\nabla}(dx_k) = -\sum\Gamma^k_{ij}dx_i\otimes dx_j$, where $\Gamma^k_{ij}$ are the Christoffel symbols of the Wagner connection $\nabla$.

Thus, in coordinates $x,w$ we have $d^s_{\nabla}(w_k) =-\sum \Gamma^k_{ij}w_iw_j$  
and the derivation $d^s_{\nabla}$ has the form:
$$
 d^s_{\nabla} =\sum w_i\partial_{x_i}-\sum\Gamma^k_{ij}w_iw_j\partial_{w_k}.
$$

\subsection{Universal differential operator} We define a total operator of third order
$$
\square_3\colon C^{\infty}(J^k\chi_3)\longrightarrow C^{\infty}(J^{k+3}\chi_3)
$$
as it was done in \cite{LY3}.

In local coordinates this operator has the form
\begin{equation}\label{UnvsTtlOprtr}
\square_3=6\sum_{\alpha, 0\le|\alpha|\le 3}\frac{u^{\alpha}}{\alpha !}\Big(\frac{d}{dx}\Big)^{\alpha},
\end{equation}
where $\alpha=(\alpha_1, \alpha_2)$, $|\alpha|=\alpha_1+\alpha_2$, $(d/dx)^{\alpha}=(d/dx_1)^{\alpha_1}(d/dx_2)^{\alpha_2}$, and $d/dx_1, d/dx_2$ are total derivatives.

\begin{theorem}\phantom{1}
  The operator $\square_3$ commutes with the action of the group $\mathcal{G}(M)$ on the jet bundles.
\end{theorem}

\subsection{Natural differential invariants of regular operators} 

By natural differential invariants of order k we mean a function on $J^k(\chi_3)$ which are $\mathcal{G}(M)$-invariant and rational along fibers of the projection $\chi_{3, k}$.

\begin{theorem}
  If $I$ is a natural differential invariant of order $\le k$ for differential operators of the third order, then $\square_3(I)$ is a  natural differential invariant of the order $\le (k+3)$ for these operators.
\end{theorem}

We say that two natural differential invariants $I_1, I_2$ are in {\it general position} if
\begin{equation}\label{GnrlPstn}
  \hat d I_1\wedge \hat d I_2\neq 0,
\end{equation}
where $\hat d$ is the total differential, and denote by $\mathcal O(I_1, I_2)\subset J^{\infty}(\chi_3)$ the open domain, where condition \eqref{GnrlPstn} holds.
\begin{theorem} Let  natural differential invariants $I_1, I_2$ are in general position and let
$$
  J^{\alpha}=\square_3(I^{\alpha}),
$$  
where $\alpha=(\alpha_1, \alpha_2)$ and $0\le|\alpha|\le 3$.

Then the field of natural differential invariants for differential operators of the third order in the domain $\mathcal O(I_1, I_2)$ is generated by invariants $I_1, I_2, J^{\alpha}$ and all their Tresse derivatives
$$
\frac{d^lJ^{\alpha}}{dI_1^{l_1}dI_2^{l_2}},
$$
where $l=l_1+l_2$.
\end{theorem}

Thus, to obtain the field of natural differential invariants of linear regular differential operators of order 3, it is enough to have two natural differential invariants $I_1$, $I_2$ in general position. 

The invariants can be obtained by various methods. 

For example, as the invariant $I_1$ one can take the free term $u_0$ of the universal operator $\square_3$, or $I_1$ is natural invariant such that $I_1(A)$ is the natural convolution of the torsion form $\theta$ of the Wagner connection and the subsymbol $\sigma_{1, A}$.

As the natural invariant $I_2$ one can take, for example, the invariant $\square_3(I_1)$.

\subsubsection{} 

Let $\mathbf{F}_k^{(3)}$ be the field of all natural differential invariants of order $\le k$ of linear scalar differential operators of order $\le 3$ on $M$.

The $\mathcal{G}(M)$-action on $M$  is transitive and, therefore, the set of all such invariants forms an $\mathbb{R}$-field $\mathbf{F} _{k}^{(3)}$.

The natural projections $\chi _{k,l}\colon\mathbf{J}^k(\chi)\rightarrow\mathbf{J}^l(\chi),\;  k>l$, define the embeddings $\mathbf{F}_l^{(3)}\hookrightarrow \mathbf{F}_k^{(3)}$ and their inductive limit $\mathbf{F}_{*}^{(3)}$ is the field of natural differential invariants of linear differential operators on the manifold $M$.

Remark that the $\mathcal{G}(M)$-action on differential operators satisfies the conditions of the Lie-Tresse theorem (see \cite{KL}) and, therefore, the field $\mathbf F_*^{(3)}$ separates regular $\mathcal G(M)$-orbits.

\section{Linear differential operators on $M\times\mathbb{R}$} 

Consider now linear differential operators $A$ of the third order on the manifold 
$\mathbf{J}^{0}\left( \pi \right) =M\times \mathbb{R}$, that satisfy the
following condition:
\begin{equation*}
\big(A-A(1)\big)(y) =0,
\end{equation*}
where $y$ is the fibrewise coordinate on $\mathbf{J}^{0}(\pi)$.

In local coordinates $(x_1, x_2)$ on $M,$ we have
representation (\ref{A}) of operators of such type:
\begin{multline*}
A=\sum\limits_{i,j,k}a_{ijk}\left( x,y\right) \frac{\partial ^{3}}{\partial
x_{i}\partial x_{j}\partial x_{k}}+\sum\limits_{i,j}a_{ij}\left( x,y\right) \frac{\partial ^{2}}{\partial
x_{i}\partial x_{j}}\\+\sum\limits_{i}a_{i}\left( x,y\right) \frac{\partial }{
\partial x_{i}}+a_{0}\left( x,y\right).
\end{multline*}

The module of these operators we will denote by $\mathit{Diff}_{3}\left( M, \mathbb{R}\right)$ and by 
$$
  \mathbb{\zeta}\colon\mathbf{Diff}_3\left( M, \mathbb{R}\right)\longrightarrow M\times\mathbb{R}
$$ 
we will denote the corresponding bundle of these operators.

As before, elements of the module of $\mathit{Diff}_{3}(M,\mathbb{R})$ are just sections of the bundle $\mathbb{\zeta }$ with the
correspondence 
$$
   \mathit{Diff}_{3}\left( M,\mathbb{R}\right)\equiv C^{\infty }(\mathbb{\zeta }),\quad A\mapsto s_A,
$$ 
where $s_{A}\left( a,y\right) =A_{\left( a,y\right) }$ and $\left( a,y\right) \in
M\times \mathbb{R}$.

The diffeomorphism group $\mathcal{G}(M) $ acts by prolongation 
$\phi \rightarrow \phi^{(0)}$ on the manifold $\mathbf{J}^0(\pi) =M\times \mathbb{R}$,  preserves the function $y$, and therefore acts in the bundle $\mathbb{\zeta }$ as well as in the $k$-jet
bundles $\mathbb{\zeta }_{k}\colon\mathbf{J}^{k}\left( \mathbb{\zeta }\right)
\rightarrow M\times \mathbb{R}$.

\subsection{Linear differential operators $\mathbf{A_f}$} Given an operator \linebreak $A\in \mathit{Diff}_3(M,\mathbb{R})$ and a function $f\in C^{\infty }(M)$ we define the operator 
$$
A_f = s_f^{\ast}\circ A\circ \pi ^{\ast}\in \mathit{Diff}_3(M) 
$$ 
as the operator that corresponds to the restriction of the section $s_{A}$ to the graph of the function $f$. 
Here $s_f\colon M\rightarrow M\times \mathbb{R}$ is the section of the bundle $%
\pi $ that corresponds to the function $f$.

In local coordinates, we get the above representation (\ref{A_f}) for this
type of operators:
\begin{multline*}
A_{f}=\sum\limits_{i,j,k}a_{ijk}\left( x,f\right) \frac{\partial
^{3}}{\partial x_{i}\partial x_{j}\partial x_{k}}+\sum\limits_{i,j}a_{ij}\left( x,f\right) \frac{\partial
^{2}}{\partial x_{i}\partial x_{j}}\\+\sum\limits_{i}a_{i}\left( x,f\right) 
\frac{\partial }{\partial x_{i}}+a_{0}\left( x,f\right).  
\end{multline*}

\subsection{Weakly nonlinear differential operators} Define now the\linebreak space of \textit{weakly nonlinear operators} of the third order $\mathit{Diff}_{3}^{w}(M)$ as the space of differential operators $A_w$ on $C^{\infty }\left( M\right) $ of the form 
\begin{equation}\label{WnOpr}
A_{w}\left( f\right)=A_{f}\left( f\right),
\end{equation}
where $A\in \mathit{Diff}_{3}\left( M,\mathbb{R}\right) $ and $f\in C^{\infty}(M)$.

\begin{proposition}\label{Prpstn7}
The mappings 
$$
\Theta\colon\mathit{Diff}_3(M,\mathbb{R})\times C^{\infty}(M)\longrightarrow\mathit{Diff}_3(M), \quad\Theta\colon(A,f)\mapsto A_f,
$$
and 
$$ 
\Theta _{w}\colon\mathit{Diff}_3(M,\mathbb{R})\longrightarrow\mathit{Diff}_{3}^{w}(M),\quad \Theta _{w}\colon A\mapsto A_w, 
$$ 
are natural in the following sense :
\begin{eqnarray*}
\phi _{\ast }\left( A,f\right) &=&\left( \phi _{\ast }^{\left( 0\right)
}\left( A\right) ,\phi _{\ast }\left( f\right) \right) , \\
\phi _{\ast }\left( A_{f}\right) &=&\big(\phi _{\ast }^{(0)}(A)\big) _{\phi _{\ast}(f)},
\end{eqnarray*}%
and%
\begin{equation*}
\phi _{\ast }(A_{w}) =\big(\phi _{\ast }^{(0)}(A)\big) _{w},
\end{equation*}%
for all diffeomorphisms $\phi \in \mathcal{G}\left( M\right)$.
\end{proposition}
\begin{proof} The proof is almost verbatim repetition of the proof of Proposition 1 in \cite{LY2w}.
\end{proof}

\section{Related pairs and their invariants}

The above proposition has the following consequences:

\begin{itemize}
\item By differential $\mathcal{G}(M)$-invariants of weakly nonlinear operators we will mean differential $\mathcal{G}(M)$-invariants of operators in\linebreak $\mathit{Diff}_3(M, \mathbb{R})$, that are $\mathcal{G}(M)$-invariant functions on jet spaces $\mathbf{J}^k(\mathbb{\zeta})$.
\item In what follows, we will require that operators $A\in \mathit{Diff}_3(M,\mathbb{R})$ under consideration are rational in $y$ (as sections of the bundle $\mathbb{\zeta}$). Also, by differential invariants of rational weakly nonlinear operators we mean $\mathcal{G}(M)$-invariant functions on jet spaces $\mathbf{J}^{k}(\mathbb{\zeta})$ that are rational along fibres of the projections 
$\pi\circ\mathbb{\zeta}_{k}\colon\mathbf{J}^{k}(\mathbb{\zeta })\rightarrow M$.
\item The group $\mathcal{G}\left( M\right) $ acts transitively on the
manifold $M$ and therefore rational differential invariants of order $\leq k$ for rational
weakly nonlinear operators form a field $\mathbf{F}%
_{k}^{w}$. We have the embedding $\mathbb{\zeta }_{k,l}^{\ast }:\mathbf{F}%
_{l}^{w}\hookrightarrow \mathbf{F}_{k}^{w},$ if $k\geq l,$ where $\mathbb{%
\zeta }_{k,l}:\mathbf{J}^{k}\left( \mathbb{\zeta }\right) \rightarrow 
\mathbf{J}^{l}\left( \mathbb{\zeta }\right) $ are the natural projections,
and $\mathbf{F}^{w}=\bigcup\limits_{k\geq 0}\mathbf{F}_{k}^{w}$ is the
field of all rational differential invariants of rational weakly nonlinear
operators.
\end{itemize}

We consider the following vector bundles over manifold $\mathbf{J}^{0}\left( \pi
\right) =M\times \mathbb{R}$:
\begin{enumerate}
\item jet bundles of functions on manifold:
\begin{equation*}
\pi _{l,0}\colon\mathbf{J}^l(\pi)\longrightarrow\mathbf{J}^0(\pi),
\end{equation*}
\item jet bundles of operators:
\begin{equation*}
\mathbb{\zeta }_k\colon\mathbf{J}^k(\mathbb{\zeta })\longrightarrow
M\times\mathbb{R},
\end{equation*}
\item the Whitney sum of vector bundles $\mathbb{\zeta }_{k}$ and $\pi _{l,0}$ 
\begin{equation*}
r_{k,l}\colon\mathbf{RP}^{k,l}=\mathbb{\zeta }_{k}\oplus _{M\times \mathbb{R}}\pi
_{l,0}\longrightarrow M\times \mathbb{R}.
\end{equation*}
 \end{enumerate}
 
We call the last bundle as the \textit{bundle of related pairs}.

Elements of the total space of this bundle are \textit{related pairs} \\$\left( [A]_{\left(
x,y\right) }^{k},[f]_{x}^{l}\right)$ consisting of $k$-jet $[A]_{\left(
x,y\right) }^{k}$ of operator $A\in \mathit{Diff}_3\left( M,\mathbb{R}\right)$ at the point $\left( x,y\right) \in M\times \mathbb{R}$ and $l$ -jet $[f]_{x}^{l}$ of function $f\in C^{\infty }(M)$ at the point $x\in M$ under condition that $f\left( x\right) =y$.

The group $\mathcal{G}(M) $ acts by prolongations in the bundles $r_{k,l}$ and by invariants of this action (or \textit{invariants of related pairs}) we mean functions on the total space 
$\mathbf{CP}^{k,l}=\mathbf{J}^k(\zeta)\oplus_{M\times \mathbb{R}}\mathbf{J}^l(\pi)$.
that are $\mathcal{G}\left(M\right) $-invariant and rational along fibers of the projection 
$\pi \circ r_{k,l}:\mathbf{RP}^{k,l}\rightarrow M$.

Because of transitivity of $\mathcal{G}(M) $-action on $M,$ all
such functions are completely determined by their values on fibre $\left(
\pi \circ r_{k,l}\right) ^{-1}(a) $ at a base point $a\in M$.
Therefore, $\mathcal{G}\left( M\right) $-invariants of related pairs form an 
$\mathbb{R}$-field $\mathbf{F}_{k,l},$ that is a subfield of the filed $\mathbf{Q}_{k,l}$ of all rational functions on the fibre $\left( \pi \circ r_{k,l}\right) ^{-1}\left( a\right)$.

The natural projections $\mathbf{CP}^{k^{\prime },l^{\prime }}\rightarrow 
\mathbf{CP}^{k,l},$ where $k\leq k^{\prime},\,l\leq l^{\prime}$, give us
embeddings of fields $\mathbf{F}_{k,l}\subset \mathbf{F}_{k^{\prime
},l^{\prime }}$,\; $\mathbf{Q}_{k,l}\subset \mathbf{Q}_{k^{\prime },l^{\prime }}$
and we define the fields $\mathbf{F}_{l},\;\mathbf{Q}_{l}$ by the inductive limits:
\begin{equation*}
\mathbf{F}_{l}=\bigcup\limits_{k\geq 0}\mathbf{F}_{k,l},\quad\mathbf{Q}
_{l}=\bigcup\limits_{k\geq 0}\mathbf{Q}_{k,l}.
\end{equation*}

Remark that $\mathbf{F}_0=\mathbf{F}^w$ is just the field of rational
differential invariants of rational weakly nonlinear differential operators.

Below we will discuss various methods of finding invariants, but first of all we remark that the vector field $\partial_{y}$ on $\mathbf{J}^{0}(\pi) =M\times \mathbb{R}$ is an invariant of the $\mathcal{G}(M)$-action. Therefore its $l$-th prolongation $\partial _{y}^{(l)}$ on 
$\mathbf{J}^{l}(\pi)$ is also $\mathcal{G}(M)$-invariant. The same is valid for the total derivation $\displaystyle\frac{d}{dy}$ that acts in 
$\mathbf{J}^{\infty}(\zeta)$. All together they define 
$\mathcal{G}\left( M\right)$-invariant derivation $\nabla$ in the fields $\mathbf{Q}_{l}$, as well as in $\mathbf{F}_{l}$, where 
\begin{equation*}
\nabla(\alpha\beta) =\frac{d\alpha}{dy}\beta +\alpha\partial_{y}^{(l)}(\beta),
\end{equation*}
$\alpha\in\mathbf{Q}_{0}$ and $\beta$ is a function on $\mathbf{J}^{l}(\pi)$.

\section{Construction of invariants}

At first we consider $y$ as a parameter and identify operators $A\in\mathit{Diff}_3(M, \mathbb{R})$ with 1-parametric family $A_y$ of operators in $\mathrm{Diff}_3(M)$.

Remark that $y$ is a $\mathcal{G}(M)$-invariant. Therefore, for any invariant 
$I\in\mathbf{F}_k^{(3)}$ of linear differential operators on manifold $M$, the function 	
$\widehat{I}\colon J^k(\zeta)\rightarrow\mathbb{R}$, where
$$	
\widehat{I}\big([A]^k_{(x, y)}\big)=I\big([A_y]^k_x\big)
$$
is a $\mathcal{G}(M)$-invariant too.

Thus we get a mapping 
$$
  \mathbf{F}^{(3)}_*\longrightarrow\mathbf{F}_0=\mathbf{F}^w,\quad I\mapsto\widehat{I},
$$ 	
that immediately gives us invariants of weakly nonlinear operators.

Moreover, application of the invariant derivation $\nabla$ essentially increases their amount.

As we have seen,  Proposition \ref{Prpstn7},  the mapping 
$$ 
\Theta _l\colon\mathbf{RP}^{l,l}\longrightarrow \mathbf{J}^l(\chi) , \quad
\Theta _l\colon ([A]_{(x,y)}^l,[f]_x^l)\mapsto [ A_f]_x^l,
$$
commutes with the $\mathcal{G}(M)$-action. Therefore, 
$$
  \Theta_l^*(I)\in\mathbf{F}_l,
$$ 
i.e., it is an invariant of related pairs for any invariant $I\in\mathbf{F}_l^{(3)}$.

\subsubsection{Descent procedure}\phantom{1}

To get invariants of weakly nonlinear operators from invariants of related
pairs we consider the following descent procedure for invariants
$$
  \mathbf{F}_{l}\longrightarrow \mathbf{F}^w.
$$ 

Let $I_{0}\in \mathbf{F}_{l}$ be an invariant of related pairs and let 
$I_{1}=\nabla(I_{0})\in\mathbf{F}_{l}, \ldots, I_{i+1}=\nabla(I_i)\in\mathbf{F}_l$ be its invariant derivatives.
Remark that the transcendence degree of the field $\mathbf{Q}_l$ over $\mathbf{%
Q}_0$ equals $N=\dim\pi _{l,0}$.
Therefore, (\cite{ZrskSml}), there are polynomial relations between rational functions ${ I_{0},...,I_{N}}$.
Denote by $J\left( I\right) \subset \mathbf{Q}_{0}[X_{0},...,X_{N}]$ the
ideal of these relations.

\begin{theorem}
Let $b_{1},..,b_{r}\in \mathbf{Q}_{0}[X_{0},...,X_{N}]$ be the reduced Gr%
\"{o}bner basis in the ideal $J(I) $ with respect to the standard lexicographic order. Then the coefficients of polynomials $b_{i}$ are natural
invariants of weakly nonlinear operators.
\end{theorem}

\begin{proof}
The action of the diffeomorphism group preserves the ideal $J\left( I\right) $
as well as the lexicographic order. The reduced Gr\"{o}bner basis in an ideal
with respect to the lexicographic order is unique (\cite{Grbnr}), and therefore, the action
preserves elements of the basis and their coefficients.
\end{proof}

\section{Example $n=1, k=3$}

Let $M=\mathbb{R}$.

Then in coordinate $x$ on M, an operator $A\in\mathrm{Diff}_3(M)$ has the form
$$ 
  A=a_3\partial^3+a_2\partial^2+a_1\partial+a_0,
$$
where the all coefficients are smooth functions on $x$ and $\partial=\partial_x$.

The symbol $a_3\partial^3$ of $A$ defines an invariant connection on the line $\R$ with the Christoffel coefficient $\Gamma$:
$$
\Gamma=-\frac{a'_3}{3a_3},
$$
the total symbol $\sigma_{(3)}$ of $A$ with respect to the Wagner connection $\nabla$ is the following:
$$ 
\sigma_{(3)}  = \sigma_3 + \sigma_2 + \sigma_1 + \sigma_0,
$$
where 
\begin{equation}\label{ScInvrnts} 
\begin{aligned}
\sigma_3=\, &a_3\partial^3,\\
\sigma_2= \, &(a_2-a'_3)\partial^2,\\
\sigma_1= \, &\big(a_1-\frac{2(a'_3)^2a_3+3a'_3-a''_3a_3}{9a_3}-(a_2-a'_3)\frac{a'_3}{3a_3}\big)\partial,\\
\sigma_0=\, &a_0
\end{aligned} 
\end{equation}
and the natural splitting of the operator $A$ is the following:
\begin{equation}\label{NtrSplit}
 A=\widehat{\sigma_{(3)}}= \widehat{\sigma_3}+\widehat{\sigma_2} + \widehat{\sigma_1}+ \widehat{\sigma_0},
\end{equation}
where the differential operators $\widehat{\sigma_k}$ are the quantizations of the symbols 
$\sigma_k$, see \cite{LY3}: 
\begin{align*}
\widehat{\sigma_3}&=a_3\partial^3+a'_3\partial^2+\frac{2(a'_3)^2a_3+3a'_3-a''_3a_3}{9a_3}\,\partial,\\
\widehat{\sigma_2}&=(a_2-a'_3)\big(\partial^2+\frac{a'_3}{3a_3}\,\partial\big),\\
\widehat{\sigma_1}&=\big(\, a_1-\frac{2(a'_3)^2a_3+3a'_3-a''_3a_3}{9a_3}-(a_2-a'_3)\frac{a'_3}{3a_3}\,\big)\partial,\\
\widehat{\sigma_0}&=a_0.
 \end{align*}

Therefore, we get from \eqref{ScInvrnts} the following $\mathcal{G}(M)$-invariants of operators $A$:
\begin{align*}
I_0&=a_0,\\
I_1&=\big\langle\,\sigma_1,\,da_0\,\big\rangle=\big(\,a_1-\frac{2(a'_3)^2a_3+3a'_3-a''_3a_3}{9a_3}-(a_2-a'_3)\frac{a'_3}{3a_3}\,\big) a'_0,\\
I_2&=\big\langle\,\sigma_2,\,da_0^2\,\big\rangle=(a_2-a'_3)(a_0')^2,\\
I_3&=\big\langle\,\sigma_3,\,da_0^3\,\big\rangle=a_3(a_0')^3.\\
\end{align*}

Now take an operator $A\in\mathrm{Diff}_3(M, \mathbb{R})$. 
\begin{align*} 
  A&=a_3(x,y)\partial_x^3+a_2(x,y)\partial_x^2+a_1(x,y)\partial_x+a_0(x,y),\\
  A_f&=a_3(x,f)\partial_x^3+a_2(x,f)\partial_x^2+a_1(x,f)\partial_x+a_0(x,f),\quad f\in C^{\infty}(M). 
\end{align*}

Then, the corresponding invariants of relates pairs are the following:
\begin{align*}
I_0=&a^0,\\
I_1 =&\Big(\,a_1-\frac{2(a_{3,x}+a{3,y}f')^2a_3+3(a_{3,x}+a_{3,y}f')}{9a_3}\\
&-\frac{\big(a_{3, xx}+2a_{3, xy}f'+a_{3, yy}(f')^2+(a_{3, x}+a_{3, y})f''\big)a_3}{9a_3}\\
&-\big(a_2-(a_{3, x}+a_{3, y}f')\big)\frac{(a_{3, x}+a_{3, y}f')}{3a_3}\,\Big)(a_{0, x}+a_{0, y}f') ,\\
I_2=&\big(a_2-a_{3, x}-a_{3, y}f'\big)(a_{0, x}+a_{0, y}f')^2,\\
I_3=&a_3(a_{0, x}+a_{0, y}f')^3.\\
\end{align*}

From the expression of $I_3$, we get 
$$
  f'=\big((\,I_3/a_3\,)^{1/3}-a_{0,x}\,\big)/a_{0,y}.
$$ 
Substituting it into $I_2$, we get a relation of the form
$$
KI_2+K_1I_3^{1/3}+K_2I_3^{2/3}+K_3I_3+K_0=0,
$$
where
$$
 \frac{K_1}{K},\quad\frac{K_2}{K},\quad\frac{K_3}{K},\quad\frac{K_0}{K}
$$
are invariants of operators $A\in\mathrm{Diff}_3(M, \mathbb{\R})$, i.e. invariants of weakly nonlinear operators.

\section{Equivalence of weakly nonlinear operators}

Let $z$ be natural differential invariant of weakly nonlinear operators and 
$A\in\mathit{Diff}_3(M,\mathbb{R})$.
 
Then the values $z(A, y_0) =s_{A}^{\ast}( z)$ is function rational in $y$ with coefficients in $C^{\infty}(M)$.  Values $z(A, y_0)$ of this function for a given value $y_0$ is a smooth function on $M$.

We say that the operator $A$ is {\it in general position} if for any point $a\in M$ there are: natural invariants $z_1, z_2$, a value $y_0$ of $y$, and a neighborhood $U\subset M$, $a\in U$, such that the mapping
$$
Z_{A, y_0}\colon U\mathcal{\longrightarrow }\,\mathbf{D}\subset\mathbb{R}^2, \quad
Z_{A, y_0}\colon x\mapsto \big( z_1(A, y_0) (x), z_2(A, y_0)(x) \big) 
$$
is a local diffeomorphism.

We say, that the mapping $Z_{A, y_0}$ is a {\it natural chart} of the operator $A$ and the functions $z_1(A, y_0), z_2(A, y_0 )$ are {\it natural coordinates} on $U$.

We call the atlas of these charts
 $\{U_\alpha, \phi_\alpha\colon U_\alpha\rightarrow\mathbf{D}_\alpha\subset\mathbb{R}^2\}$ {\it natural} if coordinates  $\phi_\alpha=\big(z_1^{\alpha}(A, y_0), z_2^{\alpha}(A, y_0) \big)$ are given by distinct invariants $\big(z_1^{\alpha}(A, y_0), z_2^{\alpha}(A, y_0)\big)\neq\big( z_1^{\beta}(A, y_0), z_2^{\beta}(A, y_0)\big)$, when $\alpha\neq\beta$.

We denote by $\mathbf{D}_{\alpha \beta}=\phi_\alpha(U_{\alpha}\cap U_{\beta})$ and we assume that domains $\mathbf{D_{\alpha}}$ and $\mathbf{D_{\alpha \beta}}$ are connected and simply connected.

Let $A_{\alpha}=\phi_{\alpha *}(A|_{U_{\alpha}})$, $A_{\alpha \beta}=\phi_{\alpha *}(A|_{U_{\alpha}\cap U_{\beta}})$ be the images of the operator $A$ in these coordinates. 

Then $\phi_{\alpha \beta *}(A_{\alpha \beta})=A_{\beta \alpha }$, where $\phi_{\alpha \beta}=\phi_{\beta}\circ\phi_{\alpha}^{-1}\colon \mathbf{D_{\alpha \beta}}\rightarrow\mathbf{D_{\beta \alpha}}$ are the transition mapping.

Take now two operators $A, A'\in \mathit{Diff}_3(M,\mathbb{R})$ and consider their natural charts $Z_{A, y_{0}}$ and $Z_{A', y_{0}^{\prime }}$. 

The $\mathcal{G}(M)$-equiva\-lence of these operators means that the local diffeomorphism  $Z_{A', y_{0}^{\prime}}^{-1}\circ Z_{A, y_{0}}$ does not depend on choice of $y_{0}$  and $y_{0}^{\prime }$. In this case we will say that these charts are \textit{coordinated}.

Summarizing we get the following result. 

\begin{theorem}
Let operators $A, A'\in\mathit{Diff}_3(M, \mathbb{R})$ be in general position, then they are 
$\mathcal{G}(M)$-equivalent if and only if the following conditions hold:
\begin{enumerate}
\item the mappings $\{\phi'_\alpha\colon U'_\alpha\rightarrow\mathbf{D}_\alpha\}$,\\ 
where $\phi'_\alpha=\big(z_1^{\alpha}(A', y'_0), z_2^{\alpha}(A', y'_0)\big)$ and  $U'_{\alpha}=(\phi'_\alpha)^{-1}(\mathbf{D}_\alpha)$,\\ constitute a natural atlas for the operator $A'$,
\item the charts $\phi_\alpha\in\{\phi_\alpha\colon U_\alpha\rightarrow\mathbf{D}_\alpha\}$ and $\phi'_\alpha\in\{\phi'_\alpha\colon U'_\alpha\rightarrow\mathbf{D}_\alpha\}$ are coordinated, 
\item $\phi'_{\alpha \beta}=\phi_{\alpha \beta}$, 
\item $A_{\alpha}=\phi'_{\alpha *}\big(A'\big|_{U'_{\alpha}}\big)$ and 
$A_{\alpha \beta}=\phi'_{\alpha *}\big(A'\big|_{U'_{\alpha\cap\beta} }\big)$.
\end{enumerate}
\end{theorem}
\begin{proof}
Any diffeomorphism $\psi\colon M\rightarrow M$ such that $\psi_*(A) = A'$ transform natural atlas to the natural one and because of $\psi^{* -1}\big(z(A)\big) = z\big(\psi_*(A)\big)$, for any natural invariant $z$, this diffeomorphism has the form of the identity map in
the natural coordinates.
\end{proof}

\begin{corollary}
Let operators $A$ and $A'\in\mathit{Diff}_3(M, \mathbb{R})$ be in general position. Then the weakly nonlinear differential operators $A_w$ and $A'_w$ are $\mathcal{G}(M)$- equivalent if and only if the operators $A$ and $A'$ are $\mathcal{G}(M)$-equivalent.
\end{corollary}



\begin{thebibliography}{99}
%
\bibitem{Grbnr}  Cox, D.; Little, J.; O'Shea, D.: {\it Ideals, Varieties, and Algorithms: An Introduction to Computational Algebraic Geometry and Commutative Algebra},  Springer, (1997), ISBN 0-387-94680-2.
%
\bibitem{KL} Kruglikov, Boris, Lychagin, Valentin, {\it  Global Lie-Tresse
theorem},  Selecta Math. (N.S.) 22 (2016), no. 3, 1357--1411.
%
 \bibitem{LY3} Lychagin, V.V., Yumaguzhin, V.A., {\it On equivalence of third order linear differential operators on two-dimensional manifolds}, Journal of geometry and physics, Vol. 146, December 2019, 103507.
 %
 \bibitem{LYk} Valentin Lychagin, Valeriy Yumaguzhin, {\it On structure of linear differential operators, acting in line bundles}, Journal of geometry and physics, Volume 148, February 2020, 103549.
%
\bibitem{LY2w} Valentin Lychagin, Valeriy Yumaguzhin, {\it Natural differential invariants and equivalence of nonlinear second order differential operators}, Journal of geometry and physics,
178(2022)104549.
%
\bibitem{Olvr} Peter J. Olver, {\it Applications of Lie groups to differential equations}, Springer-Verlag
%
\bibitem{Wgnr} Wagner,V.V., {\it Two dimensional space with cubic metric}, Scientific notes of Saratov State University, Vol. 1(XIV),   Ser.  FMI, No. 1, 1938. (in Russian).
%
\bibitem{ZrskSml} Zariski O., Samuel P.:{\it Commutative Algebra}, vol. 1, 2, Van Nostrand, 1960.
%
\end{thebibliography}
\end{document}